\documentclass[journal]{IEEEtran}
\usepackage{booktabs, multicol, multirow}

\usepackage{cite}

\ifCLASSINFOpdf
   \usepackage[pdftex]{graphicx}
   \graphicspath{{../pdf/}{../jpeg/}}
  \DeclareGraphicsExtensions{.pdf,.jpeg,.png}
\else
\fi
\usepackage{amsmath}
\usepackage{amsthm}
\usepackage{amssymb}
\usepackage{flexisym}
\DeclareMathOperator*{\mindot}{min.}
\DeclareMathOperator*{\maxdot}{max.}

\DeclareMathOperator*{\argmin}{argmin}
\DeclareMathOperator*{\prox}{prox}
\newtheorem{prop}{Proposition}
\newtheorem{remark}{Remark}

\theoremstyle{definition}
\newtheorem{definition}{Definition}
\usepackage{comment}

\usepackage{algorithm}
\usepackage{algorithmic}

\usepackage{array}
\newcolumntype{C}[1]{>{\centering\let\newline\\\arraybackslash\hspace{0pt}}m{#1}}

\usepackage{color}
\definecolor{dblue}{rgb}{0,0,0.8}

\usepackage{xcolor}

\usepackage{subfig}
\usepackage{graphicx}
\usepackage{fixltx2e}

\usepackage{stfloats}
\usepackage{nomencl}
\usepackage{ifthen}
\renewcommand{\nomgroup}[1]{%
\ifthenelse{\equal{#1}{A}}{\item[\textbf{\textit{Acronyms:}}]}{\vspace{0.15cm}%
\ifthenelse{\equal{#1}{B}}{\item[\textbf{\textit{Sets:}}]}{\vspace{0cm}%
\ifthenelse{\equal{#1}{C}}{\item[\textbf{\textit{Constant Parameters:}}]}{\vspace{0cm}%
\ifthenelse{\equal{#1}{D}}{\item[\textbf{\textit{Transmission Network Variables:}}]}{\vspace{0cm}%
\ifthenelse{\equal{#1}{E}}{\item[\textbf{\textit{Distribution Network i Variables:}}]}{\vspace{0cm}%
\ifthenelse{\equal{#1}{F}}{\item[\textbf{\textit{Boundary Variables:}}]}{\vspace{0cm}%
\ifthenelse{\equal{#1}{G}}{\item[\textbf{\textit{Indices:}}]}{\vspace{0cm}%
}}}}}}}
}
\makenomenclature

\usepackage[draft]{todonotes}

\begin{document}
%
\title{Scalable Distributed Non-Convex ADMM-based Active Distribution System Service Restoration}
%
%
%

\author{Reza~Roofegari nejad,~\IEEEmembership{Student Member,~IEEE,}
        and~Wei~Sun,~\IEEEmembership{Member,~IEEE}
\thanks{R. Roofegari nejad, and W. Sun are with the Department of Electrical and Computer Engineering, University of Central Florida, Orlando, FL. 32816 USA. (e-mail: rezarn@ece.ucf.edu, sun@ucf.edu).}
}

\maketitle

\begin{abstract}
Distributed restoration can harness distributed energy resources (DER) to enhance the resilience of active distribution networks. However, the large number of decision variables, especially the binary decision variables of reconfiguration, bring challenges on developing effective distributed distribution service restoration (DDSR) strategies. This paper proposes a scalable distributed optimization method based on the alternating direction method of multipliers (ADMM) for non-convex mixed-integer optimization problems, and applies to develop the DDSR framework. The non-convex ADMM method consists of relax-drive-polish phases, 1) relaxing binary variables and applying the convex ADMM as a warm start; 2) driving the solutions toward Boolean values through a proximal operator; 3) fixing the obtained binary variables to polish continuous variables for a high-quality solution. Then, an autonomous clustering strategy together with consensus ADMM is developed to realize the distributed cluster-based framework of restoration. The non-convex ADMM-based DDSR can determine DER scheduling and switch status for reconfiguration and load pickup in a distributed manner, energizing the out-of-service area from local faults or total blackouts in large-scale distribution networks. The effectiveness and scalability of the proposed DDSR framework are demonstrated through testing on the IEEE 123-node and IEEE 8500-node test feeders.
\end{abstract}


\nomenclature[A,13]{TN}{Transmission Network}
\nomenclature[A,05]{CB}{Capacitor Bank}
\nomenclature[A,05]{DN}{Distribution Network}
\nomenclature[A,08]{DSR}{Distribution System Restoration}
\nomenclature[A,08]{DDSR}{Distributed Distribution System Restoration}
\nomenclature[A,11]{NC-ADMM}{Non-Convex ADMM}
\nomenclature[A,02]{ADN}{Active Distribution Networks}
\nomenclature[A,03]{DER}{Distributed Energy Resources}
\nomenclature[A,04]{DG}{Distributed Generator}
\nomenclature[A,10]{MT}{Microturbine}
\nomenclature[A,11]{PV}{Photovoltaic}
\nomenclature[A,12]{SLA}{Smart Local Agent}
\nomenclature[A,17]{VR}{Voltage Regulator \vspace{0.05cm}}
\nomenclature[A,01]{ADMM}{Alternating Direction Method of Multipliers}

\nomenclature[B,01]{$\mathcal{T}$}{Set of restoration time steps}
\nomenclature[B,02]{$\mathcal{N}, \mathcal{N}^{\text{sub}}, \mathcal{K}$}{Set of all nodes, substations, and clusters}
\nomenclature[B,02]{$\mathcal{L}, \mathcal{E}, \mathcal{G}$}{Set of loads, lines, and DERs}
\nomenclature[B,03]{$\mathcal{E}^S, \mathcal{E}^F$}{Set of switchable and faulted lines}
\nomenclature[B,03]{$\phi$}{Set of three phases as $\{a, b, c\}$}
\nomenclature[B,04]{$\pi(i)$}{Set of parent nodes of node $i$}
\nomenclature[B,05]{$\delta(i)$}{Set of children nodes of node $i$}



%
\IEEEpeerreviewmaketitle

\section{Introduction}
%
%
%
%
\IEEEPARstart{S}{mart} grid technologies have been applied to enhance the resilience of distribution networks (DN); however, faults and outages are still inevitable due to challenges from natural disasters or man-made attacks \cite{Bie2017Battling}. It is critical to effectively respond to extreme events and optimally restore the electric service \cite{Yao2020Rolling}. For example, after an outage, microgrids can be utilized to supply critical loads \cite{Xu2018Microgrids} or maximize restored loads \cite{Wang2015Self-Healing}; distributed energy resources (DER) can be harnessed to energize islanded microgrids \cite{Chen2018Sequential}; and advanced communication and control devices, such as remote-controlled switches, provide great potentials for advanced restoration strategies \cite{Pathirikkat2018Protection}. \par

Distribution service restoration (DSR) aims to restore maximum out-of-service loads through finding available paths by switching operation and picking up loads after a blackout \cite{Roofegari2019Distributed}. Large-scale DNs and the increasing penetration of DERs make DSR one of the most complicated and challenging problems of active distribution networks (ADN) \cite{Marques2018Service}. The prevailing centralized restoration scheme requires a powerful centralized controller to communicate with components, collect data, carry out large-scale sophisticated computations, and send out commands \cite{Li2020Decentralized}. However, centralized infrastructures are costly, might suffer from single-point failures, and are limited by information privacy of entities \cite{Roofegari2019Distributed}. \par

Distributed optimization and control is a promising solution to address the aforementioned issues, but most of the methods require convexity to guarantee the convergence, which is difficult to meet due to the mixed-integer characteristic of many power system problems including DSR \cite{Roofegari2019Distributed}. For example, the multi-agent based distributed service restoration scheme has been developed to address the restoration problem \cite{Elmitwally2015Fuzzy, Abel2018Decentralized, Li2020Decentralized}. However, these methods require agents to access all information against data privacy, and usually neglect the scalable problem formulation for large-scale DNs. Authors in \cite{Shen2019Distributed} proposed a distributed restoration framework based on the alternating direction method of multipliers (ADMM). However, simply projected ADMM to deal with binary variables usually obtains infeasible solutions \cite{Boyd2011Distributed}, and the proposed clustering structure can not include switches within clusters.\par

The ADMM method was originally developed for convex optimization problems, but it turns out to be a powerful heuristic method even for non-convex (NC) problems \cite{Takapoui2017simple}. Recently, it has gained a lot of attention to find the approximate solution for NP-hard problems \cite{Diamond2018general}. Although there is no guarantee for ADMM-based methods converge to a global solution for NC problems, it is an attempt to find a high-quality local solution, considering the difference between global and local optimal solutions in practice may not be significant\cite{Diamond2018general}. Accordingly, \cite{Boyd2011Distributed} and \cite{Diamond2018general} proposed a heuristic projection method based on ADMM for general NC problems, which usually yields infeasible or low-quality local solutions. Therefore, in this paper, a distributed optimization method based on ADMM is proposed to achieve a high-quality solution for mixed-integer optimization problems. Different from \cite{FEIZOLLAHI2015Large-scale} that locks binary variables during convergence, the proposed method tries to drive them simultaneously toward Boolean values. The developed NC-ADMM method, together with an autonomous clustering scheme, is applied to the DSR problem. \par

The main contributions of this paper are threefold: (1) Developed a cluster-based distributed distribution service restoration (DDSR) framework including reconfiguration and load pickup for unbalanced large-scale distribution networks, based on the proposed heuristic NC-ADMM algorithm. The DDSR framework establishes clusters and solves the DSR problem through solving small subproblems for each smart local agent (SLA) and exchanging limited information between adjacent clusters. (2) Developed a heuristic distributed approach based on ADMM through relax-drive-polish phases for NC mixed-integer problems, and applied to the DDSR problem for a high-quality suboptimal solution. (3) Developed an autonomous two-stage clustering strategy to enhance the scalability of NC-ADMM-based DDSR for large-scale distribution networks. \par

The remainder of the paper is as follows. Sections II and III introduce the distributed scheme for DSR problem and DSR problem formulation, respectively. Section IV presents the proposed autonomous clustering strategy for large-scale distribution networks. Section V introduces the proposed NC-ADMM method with the application on the DDSR problem. Finally, section VI demonstrates numerical results and analysis, and section VII concludes the paper. \par

\section{DDSR Framework}\label{DDSR Scheme}
DSR can be categorized into emergency service restoration and blackout service restoration \cite{Roofegari2019Distributed}. The emergency service restoration aims to isolate faulted areas by opening sectionalizing switches and to energize unfaulted out-of-service areas by closing tie-switches. For blackout service restoration, the entire DN is out of service, and the restoration can be accomplished through multiple time steps, depending on the bulk power system restoration procedure. In this paper, the proposed DDSR can handle both top-down and bottom-up restoration strategies, aiming to promptly restore as much load as possible in areas where electric service is disrupted.\par

The cluster-based DDSR framework is shown in Fig. \ref{Fig. DDSR Framework}. The distribution network is divided into multiple agents, and these agents can be realized by any node, DER, or smart switch. Rather than controlling each node independently as in our previous work \cite{Roofegari2019Distributed}, which requires many intelligent entities, a cluster of nodes are controlled by assigning an SLA. Each SLA monitors, controls, and dispatches all the loads, DERs, switches, and capacitor banks (CB) within its cluster, and also exchanges data with neighboring clusters through a two-way communication network. Neighboring clusters are those connected through power delivery elements such as distribution lines. The cluster-based structure improves the scalability and convergence speed of the proposed distributed algorithm and also enables the practical implementation of distributed communication and control. It is worth mentioning that each SLA is responsible for its own territory and the boundaries of clusters are artificial. Moreover, the boundary elements could be controlled by all related SLAs when the proposed algorithm reaches consensus on the operation.\par
\begin{figure}
\centering
	\includegraphics[width=3.6in]{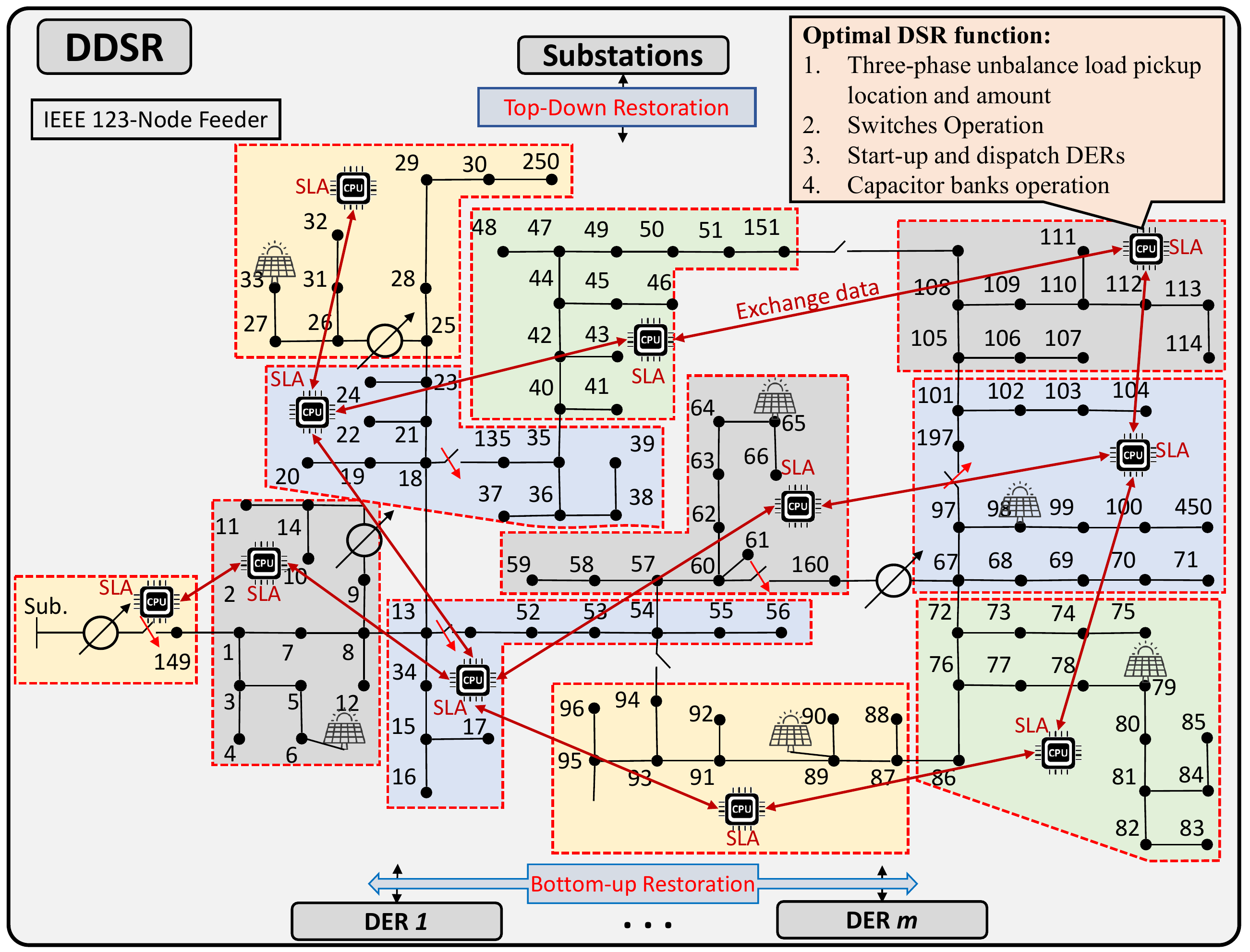}
	\caption{The proposed DDSR framework.}
	\label{Fig. DDSR Framework}
\end{figure}
The proposed heuristic consensus NC-ADMM method for DDSR can be realized through solving a small subproblem by each SLA and exchanging limited information among neighboring clusters as shown in Fig. \ref{Fig. DDSR Framework}. Each subproblem includes various constraints, such as power flow and security constraints, to guarantee the solution feasibility for the restoration procedure. When faults occur, the protection system detects and isolates faults thorough sectionalizing switches. Then, depending on the type of event, the restoration procedure starts and might last from one to several time steps, each of which has enough duration for operation and stabilization of all involved components. In the case of a bulk system outage, the steps are imposed by the generation capability curve in transmission system, i.e., available power at each distribution substation. Furthermore, the restoration might include the energization of unfaulted out-of-service areas through closing tie-switches. Due to the protection system limitation, the tree topology of the entire DN must be preserved during restoration, which is a challenging constraint for DDSR \cite{Taylor2012convex}.\par

Accordingly, the optimal restoration plan of DDSR can be achieved in the following iterative steps:
\subsubsection{DSR subproblems step} The local DSR subproblems are solved by each SLA for the entire restoration span. A decomposable Lagrange function is developed for the whole restoration problem and then decomposed into subproblems for each agent. These subproblems can be solved independently by each SLA while consensus variables are being fixed in this step. Then, each agent determines its load pickup, switching, generation dispatch, CB operation, voltage, and power flow for all loads and power delivery elements including inner and boundary ones. It is worth mentioning that in this step all binary variables are relaxed and confined into a convex hull; therefore, all subproblems are convex optimization problems. Also, the Boolean value for these variables is achieved through the relax-drive-polish phases of the NC-ADMM method. 
\subsubsection{Consensus step} All neighboring agents exchange information related to the voltage, power flow, lines energization, and radiality constraints, and update the consensus variables.    
\subsubsection{Lagrange multipliers update step} The Lagrange multipliers related to the boundary and binary variables are updated by each SLA based on results from the previous two steps and exchanged data. Updating these multipliers enforces general constraints for the entire distribution network, including supplying critical loads in other clusters, radial topology, and total power balance.

The iterative procedure of DDSR continues until satisfying convergence criteria or time limits. Then, the algorithm provides a sequence of control actions for each cluster to restore its loads based on their priority, and operate switches and various components in the network within each time step of the total restoration span. \par

\section{DSR Problem Formulation}\label{PROBLEM FORMULATION}
\subsection{Model Respresentation}
This section provides DSR formulation for an unbalanced DN consisting of $N$ nodes denoted by set $\mathcal{N}:=\{1,...,N\}$. Also, $\mathcal{L}$, $\mathcal{E}$, and $\mathcal{G}$ show the set of loads, lines, and DERs, respectively. Superscripts $S$ and $F$ represent subsets of switchable lines and faulted lines, as $\mathcal{E}^S, \: \mathcal{E}^F \subseteq \mathcal{E}$; and superscripts $d$ and $nd$ stand for subsets of dispatchable and non-dispatchable loads as $\mathcal{L}^d \cup \mathcal{L}^{nd}= \mathcal{L}$. Furthermore, $\mathcal{T}$ denotes the set of restoration time steps as $t \in \mathcal{T}$, and $\phi$ stands for three phases of $\{a, b, c\}$. The inner product of two vectors $\vec{x},\vec{y}$ is denoted by $\langle \vec{x},\vec{y}  \rangle$, and the element-wise product and division are denoted by $\odot$ and $\oslash$ respectively.\par

Three-phase restored load $i$ at time $t$ is represented by $\vec{P}_{i,t}^{L}+j\vec{Q}_{i,t}^{L} \in \mathbb{R}^{3\times 1}$. Let $\vec{P}_{i,t}^{\text{G}}+j \vec{Q}_{i,t}^{\text{G}} \in \mathbb{R}^{3\times 1}$ denote three-phase power generation of DER$_i$, and $\vec{P}_{ij,t}+j\vec{Q}_{ij,t} \in \mathbb{R}^{3\times 1}$ and $\textbf{\textit{r}}_{ij}+j\textbf{\textit{x}}_{ij}=\textbf{\textit{z}}_{ij} \in \mathbb{R}^{3\times 3}$ denote the power flow and impedance of distribution line $ij$ between nodes $i$ and $j$. In order to capture the energization status of load $i$ and bus $i$, and the connection status of line $ij$ at time $t$, a set of binary variables are defined as $x_{i,t}^L$, $x_{i,t}^B$, and $\alpha_{ij,t}$. Nevertheless, to simplify the DSR problem formulation and solution algorithm, binary variables $\alpha_{ij}$ are only assigned for switchable lines. \par
\vspace{-0.2cm}
\subsection{Objective Function}
The DSR problem is modeled as a mixed-integer convex optimization problem. A multi-objective function is developed to maximize total restored loads at all time steps considering load priorities, as well as the optimal operation of switches. Binary variables of switch status are prioritized so that the normally-closed switches have a higher priority compared to the normally-open ones. The motivation is to keep network topology close to the one under the normal operation. These two terms are regularized through $c_1$ and $c_2$ as below:  
\begin{equation}\label{equ. objective function}
       \maxdot \sum_{t \in \mathcal{T}} \; (c_1 \sum_{i \in \mathcal{L}} \left \langle \vec{w}_i^L , \vec{P}_{i,t}^{L}  \right \rangle  + \; \; c_2 \sum_{(i,j) \in \mathcal{E}^{S}} w_{ij}^{S} \; \alpha_{ij,t}) 
\end{equation}
where $\vec{w}^L_i$ and $w_{ij}^S$ denote priority of load $i$ and switchable line $ij$.
\vspace{-0.3cm}
\subsection{Constraints}
\subsubsection{Security constraints}
Security constraints guarantees the operational limits of various elements in DN. Constraints \eqref{equ. DSR non-dispatchable load pickup} and \eqref{equ. DSR dispatchable load pickup} represent the load pickup capacity of non-dispatchable and dispatchable loads. Nodal voltage and branch power flow should be within the limit for energized nodes and lines, as in \eqref{equ. DSR voltage limit} and \eqref{equ. DSR branch power capacity2} where $\vec{U}_{i,t}= |\vec{v}_{i,t}|^{\odot2}$ is the square of voltage magnitude at node $i$. Constraint \eqref{equ. DSR substation capacity} demonstrates the restorative capacity of transmission network ($P_t^{\text{sub}}+jQ_t^{\text{sub}}$), as $P_{1j,t}^{\phi}+jQ_{1j,t}^{\phi}$ is the total three-phase power flow of the connected line to the substation. The output of energized CB is limited by its maximum capacity in \eqref{equ. DSR Capacitor bank}. A convex quadratic constraint is used to limit the output power of inverter-based photovoltaic (PV) generators in \eqref{equ. DSR PV generation limitation} which are the only type of DERs modeled in this paper. Finally, Voltage regulators (VR) are assumed to be wye-connected type B, which the voltage between primary and secondary sides are represented by \eqref{equ. DSR Voltage regulator}. This equation can be linearized based on \cite{Robbins2016Optimal}; however, for simplicity, VRs are assumed to be constant during restoration.
\begin{equation}\label{equ. DSR non-dispatchable load pickup}
\vec{P}_{i,t}^{L}=x_{i,t}^{L} \; \vec{P}_{i}^{L,\text{max}}, \quad \vec{Q}_{i,t}^{L}=x_{i,t}^{L} \; \vec{Q}_{i}^{L, \text{max}}, \; \forall i \in \mathcal{L}^{nd}
\end{equation}
\begin{equation}\label{equ. DSR dispatchable load pickup}
0 \leq \vec{P}_{i,t}^{L}\leq \vec{P}_{i}^{L, \text{max}}, \quad  0 \leq \vec{Q}_{i,t}^{L}\leq \vec{Q}_{i}^{L, \text{max}}, \; \forall i \in \mathcal{L}^{d}
\end{equation}
\begin{equation}\label{equ. DSR voltage limit}
x_{i,t}^B{(V^{min})}^2\leq \vec{U}_{i,t}(ph)\leq x_{i,t}^B{(V^{max})}^2, \; \forall ph \in \phi
\end{equation}
\begin{equation}\label{equ. DSR branch power capacity2}
(\vec{P}_{ij,t})^{\odot2}+(\vec{Q}_{ij,t})^{\odot2}\leq \alpha_{ij,t} \; (\vec{S}_{ij}^{\text{max}})^{\odot2}
\end{equation}
\begin{equation}\label{equ. DSR substation capacity}
0 \leq P_{1j,t}^{\phi} \leq P_{t}^{\text{sub}} , \quad 
0 \leq Q_{1j,t}^{\phi} \leq Q_{t}^{\text{sub}}
\end{equation}
\begin{equation}\label{equ. DSR Capacitor bank}
0 \leq \vec{Q}_{i,t}^{\text{cap}}\leq x_{i,t}^B \; \vec{Q}_{i}^{\text{cap, max}}
\end{equation}
\begin{equation}\label{equ. DSR PV generation limitation}
(\vec{P}_{i,t}^{G})^{\odot2}+(\vec{Q}_{i,t}^{G})^{\odot2}  \leq x_{i,t}^B \; (\vec{S}_{i}^{\text{Inv, max}})^{\odot2}
\end{equation}
\begin{equation}\label{equ. DSR Voltage regulator}
\vec{U}_{i,t}=\vec{a}^{\odot2} \odot \vec{U}_{j,t}\, ; \quad \vec{a}=1+ {R_i}\% \cdot (\vec{n}_{i,t}^{\text{tap}}\oslash\overline{n}_{i}^{\text{tap}})
\end{equation}
\subsubsection{Power flow constraints}
Power flow equations are based on the branch flow model for three phase unbalanced DNs in \cite{Roofegari2019Distributed}. For simplicity, these equations are linearized by removing the quadratic loss term \cite{Robbins2016Optimal}. The power flow equations include the linearized voltage drop for node $i$ in \eqref{equ. DSR voltage drop}, and active and reactive power balance in \eqref{equ. Active power balance} and \eqref{equ. Reactive power balance}, where $\widetilde{\textbf{\textit{r}}_{ji}}$ and $\widetilde{\textbf{\textit{x}}_{ji}}$ are unbalanced line impedance matrices, as referred in \cite{Roofegari2019Distributed}, and $j$ and $\delta(i)$ are parent node and set of children nodes of node $i$. For switchable lines, the big $M$ method can be applied on equation \eqref{equ. DSR voltage drop} to enable it based on the connectivity. 
\begin{subequations}\label{equ. DSR Power flow equations}
\begin{equation} \label{equ. DSR voltage drop}
\vec{U}_{i}=\vec{U}_{j}-2\left ( \widetilde{\textbf{\textit{r}}_{ji}} \vec{P}_{ji}+\widetilde{\textbf{\textit{x}}_{ji}} \vec{Q}_{ji} \right ) \; \forall(j,i) \in \mathcal{E}  \setminus \mathcal{E}^S 
\end{equation}
\vspace{-0.3cm}
{\begin{equation}\label{equ. Active power balance}
\vec{P}_{ji,t}=  \vec{P}_{i,t}^{{L}}+\sum_{m \in \delta(i)}\vec{P}_{im,t}-\vec{P}_{i,t}^{{G}} 
\end{equation}}
\vspace{-0.25cm}
\begin{equation}\label{equ. Reactive power balance}
\vec{Q}_{ji,t}= \vec{Q}_{i,t}^{{L}}+\sum_{m \in \delta(i)}\vec{Q}_{im,t}-\vec{Q}_{i,t}^{{G}}-\vec{Q}_{i,t}^{\text{cap}} 
\end{equation}
\end{subequations}

\subsubsection{Topological, connectivity and sequencing constraints}
These constraints guarantees radial operation, isolation of fault, and sequential restoration of DNs. Load shedding is prohibited by \eqref{equ. DSR sequencing load} for energized loads during restoration. If a fault happens, the protection system can locate and isolate the faulted area through sectionalizing switches. Then, by forcing the related binary variables of faulted lines to be zero as \eqref{equ. DSR fault isolation}, the faulted area remains isolated during the DSR procedure. The radial configuration of DN is guaranteed through spanning tree constraints \eqref{equ. Spanning Tree constraints} \cite{Taylor2012convex}. Two auxiliary binary variables $\beta_{ij}$ and $\beta_{ji}$ are associated with each line $ij$, denoting the direction of flow if any. Then, equations \eqref{equ. DSR parent-child line} and \eqref{equ. DSR parent-child switchable line} indicate that for each line $ij$, either node $j$ is the parent of node $i$ ($\beta_{ij}=1$), or node $i$ is the parent of node $j$ ($\beta_{ji}=1$). Also, equation \eqref{equ. DSR EachNode one parent} requires every node other than substation has exactly one parent node, while substation does not have any parent.
\begin{equation}\label{equ. DSR sequencing load}
\vec{P}_{i,t}^{L} \geq \vec{P}_{i,t-1}^{L}, \quad \forall i \in \mathcal{L}
\end{equation}
\begin{equation}\label{equ. DSR fault isolation}
    \alpha_{ij,t}=0,\; x_{i,t}^B=0, \; x_{j,t}^B=0 \; \; (i,j)\in \mathcal{E}^F
\end{equation}{}
\begin{subequations}\label{equ. Spanning Tree constraints} {\begin{equation}\label{equ. DSR parent-child line}
\beta_{ij,t}+\beta_{ji,t} = 1, \; \; \forall (i,j) \in \mathcal{E} \setminus \{ \mathcal{E}^{F} \cup \mathcal{E}^{S} \}
\end{equation}}
\vspace{-0.25cm}
{\begin{equation}\label{equ. DSR parent-child switchable line}
\beta_{ij,t}+\beta_{ji,t} = \alpha_{ij,t}, \; \; \forall (i,j) \in  \mathcal{E}^S
\end{equation}}
\vspace{-0.25cm}
{\begin{equation}\label{equ. DSR EachNode one parent}
 \sum_{\forall i \in \mathcal{N}} \beta_{ij,t} \leq 1,  \; \forall (i,j)  \in \mathcal{E}, \quad \beta_{ij,t} = 0,  \;  j \in {\mathcal{N}}^{sub} 
 \end{equation}}
\end{subequations}

\section{Autonomous Clustering Algorithm}\label{Section Clustering}
A two-level autonomous clustering strategy is developed to maximize the convergence speed and minimize the information exchange. The first level is to identify the optimal number and size of clusters, and the second level is to detect and generate clusters within a DN. This self-organizing strategy is applicable for various communication typologies, different types of information sharing, and a large number of nodes. 
\begin{prop}
The optimal number of clusters in a DN to maximize the convergence speed and minimize the information exchange is achieved by square root of nodes. 
\end{prop}
\begin{proof}
Considering a DN with $N$ nodes divided into $k$ clusters each with $m$ nodes, an optimization problem is formulated with the objective function of minimizing $\lambda_1 k+\lambda_2 m$ subject to $km=N$, where $\lambda_1$ and $\lambda_2$ are weighting factors for a trade-off between communication requirement and computational burden of each cluster for the convergence speed. Substituting the constraint as $m=N/k$ and assuming $\lambda_1=\lambda_2=1$, yields the problem as minimizing $k+N/k$. Then, by taking the derivative, the objective value obtains as $k=\sqrt{N}$ and subsequently $m=\sqrt{N}$.
\end{proof}
In second level, a bottom-up travers method is developed for DNs with tree structure. This method starts from the leaves of the tree and ends by the root node, as provided in Algorithm \ref{Alg. Bottom-up travers}. First, the algorithm assigns a weight to each node demonstrating its subtree size. Next, the algorithm finds the nodes with the weight close to the ideal number and detaches them as a cluster. This procedure is repeated until all nodes are removed from the tree structure and been divided into clusters. \par
\begin{algorithm}[]
 \caption{Bottom-up travers clustering algorithm}
 \begin{flushleft}
 \hspace*{\algorithmicindent} \textbf{Input:} $T$: The network tree, $N$: number of nodes, \\ \hspace{0.5cm}$k$: desired number of clusters, $r$: relaxation factor; \\
 \hspace*{\algorithmicindent} \textbf{Output:} $T\textprime$: Clustered network tree.
 \end{flushleft}
\begin{algorithmic}[1]
 \STATE \textbf{Calculate} the number of nodes in each cluster ($m$)
 \STATE \textbf{Traverse} the tree in the bottom-up postorder manner
 \STATE \textbf{for} each node $v$  \textbf{Assign} its subtree size
 \STATE \textbf{Find} all nodes $v\textprime$ with subtree size $=m \pm r$
 \STATE \textbf{Cluster} each node $v\textprime$ with its subtree
 \STATE \textbf{Remove} all clustered nodes
 \STATE \textbf{Repeat} Algorithm \ref{Alg. Bottom-up travers} until all nodes been clustered
\end{algorithmic}
 \label{Alg. Bottom-up travers}
\end{algorithm}
\vspace{-0.3cm}
\begin{remark}
Finding $\sqrt{N}$ clusters with $\sqrt{N}$ nodes is a hard problem and might be impossible due to the topology of the network. Accordingly, the proposed algorithm provides a high quality solution even for large-scale DNs, in a swift manner. 
\end{remark}

\vspace{-0.5cm}
\section{Distributed Solution Methodology}\label{Distributed Algorithm DDSR}
\subsection{Proposed Distributed Algorithm of Non-convex ADMM} \label{Distributed Algorithm of Non-convex ADMM}
This paper developed a heuristic approach, including the \textit{relax-drive-polish} procedure, to enhance the performance of heuristic ADMM for non-convex problems. It solves the relaxed convex problems in each iteration of the ADMM, and assigns auxiliary variables to force relaxed binary variables toward Boolean values during convergence. It also combines with the consensus ADMM \cite{Boyd2011Distributed}, which provides a parallel computational framework for all agents to reach consensus on binding variables through limited information exchange.\par
\begin{definition}\label{Def. binding variables}
Binding variables are those involved in binding constraints among agents (SLA) in a decomposed problem.
\end{definition}

The optimization problems of DSR with decomposable objective function for $\mathcal{K}$ clusters can be generalized as following:
\begin{equation}\label{equ. relaxed general optimization problem}
    \begin{aligned}
    \mindot_{\vec{x}, \vec{z}} &\: \sum_{i\in \mathcal{K}} f_i(\vec{x}_i,\vec{z}_i)\\
    s.t. & \; \; \vec{x}, \vec{z} \in C; \; \vec{z} \in S
    \end{aligned}
\end{equation}
where $\vec{x}\in \mathbb{R}^n$ and $\vec{z} \in \mathbb{R}^q$ are decision variables with a convex objective function, and inequality and equality constraints define a convex set $C$ for $\vec{x}$ and $\vec{z}$. Respectively, decision variables can be classified for each cluster $i$ as $[\vec{x}_i,\vec{z}_i]$. Based on Definition \ref{Def. binding variables}, these variables can be categorized into binding and interior variables as $[\vec{x}_i^{Bi},\vec{x}_i^{In},\vec{z}_i^{Bi},\vec{z}_i^{In}]$. Furthermore, $S$ represents the non-convex set, which for DSR problem is Boolean set as $S=\{0,1\}^q$. The relaxation of Boolean constraints is achieved through introducing auxiliary variables of $\vec{y}$ within a convex hull and making them equal to $\vec{z}$ as a consensus constraint $\vec{y}=\vec{z}, \; 0 \leq \vec{y} \leq 1$.\par
Through introducing augmented Lagrangian over consensus constraint in \eqref{equ. relaxed general optimization problem}, the ADMM solution procedure can be formed as \eqref{equ. generalized projection Non-convex ADMM} where $\vec{\mathfrak{u}}$ is the Lagrange multiplier for ADMM in scaled form. However, the convergence procedure is unsmooth and the results are usually infeasible. Therefore, a non-convex ADMM method is developed next.
\begin{subequations}\label{equ. generalized projection Non-convex ADMM}
 \begin{align}{}
 \begin{split}
 (\vec{x},\vec{y})^{(k+1)} &:= \argmin_{\vec{x},\vec{y} \in C, 0\leq\vec{y}\leq1} \bigg( \sum_{i \in \mathcal{K}}f_i(\vec{x}_i,\vec{y}_i)+\\&
 \hspace{2cm} \frac{\rho}{2} \|\vec{y}-{\vec{z}}^{(k)}+{\vec{\mathfrak{u}}}^{(k)}\|_2^2\bigg)
 \label{equ. step 1 generalized projection Non-convex ADMM}
 \end{split}{}\\
 {\vec{z}}^{(k+1)} &:= \Pi_{S} ({\vec{y}}^{(k+1)}+{\vec{\mathfrak{u}}}^{(k)}) \label{equ. step 2 generalized projection Non-convex ADMM} \\
 {\vec{\mathfrak{u}}}^{(k+1)} &:= {\vec{\mathfrak{u}}}^{(k)} + {\vec{y}}^{(k+1)} - {\vec{z}}^{(k+1)} \label{equ. step 3 generalized projection Non-convex ADMM}
 \end{align}
\end{subequations}{}
\begin{remark}
The $\Pi_S$ stands for the projection over set $S=\{0,1\}^q$ which is given by rounding the entries to $0$ or $1$.
\end{remark}

\subsubsection{Step 1-Relax} 
Inspired by the Douglas-Rachford splitting method \cite{Eckstein1989Splitting}, instead of projection in \eqref{equ. step 2 generalized projection Non-convex ADMM}, the proximal operator, as defined in \eqref{equ. Proximal operator definition}, is applied to drive $\vec{z}$ toward Boolean values as \eqref{equ. prox ADMM binary update}, where $\tilde{t}$ is the regularization factor. In order to drive $\vec{z}$ toward the Boolean values, it is proposed to establish $I(\vec{z}):=\frac{1}{2}\|\vec{z}-\Pi_S(\vec{w})\|^2_2$.
\begin{equation} \label{equ. Proximal operator definition}
    \prox_{\tilde{t}, I(\vec{z})}(\vec{w}) := \argmin_{\vec{z}} \; \frac{1}{2\tilde{t}}\|\vec{z}-\vec{w}\|^2_2+I(\vec{z})
\end{equation}
\begin{equation}\label{equ. prox ADMM binary update}
\vec{z}^{(k+1)} := \prox_{\tilde{t},I(\vec{z})} \bigg(\vec{y}^{(k+1)}+\vec{\mathfrak{u}}^{(k)} \bigg):=\prox_{\tilde{t}, I(\vec{z})}(\vec{w})
\end{equation}
Then, the optimization problem of \eqref{equ. Proximal operator definition} is a compromise to choose $\vec{z}$ between the Boolean value and the consensus value of $\vec{w}$. The closed-form solution of \eqref{equ. Proximal operator definition} is obtained as \eqref{equ. closed form solution for z}. When $\tilde{t}=0$, $\vec{z}=\vec{w}$ as the consensus value; when $\tilde{t}$ moves toward infinity, $\vec{z}$ is forced to take the projected value of $\vec{w}$; when $\tilde{t}$ becomes large enough, $\vec{z}$ might achieve Boolean value. 
\begin{equation}\label{equ. closed form solution for z}
    {\vec{z}}^{(k+1)} = (1+\tilde{t}^{(k)})^{-1}\bigg(\vec{w}+\tilde{t}^{(k)}\Pi_S(\vec{w})\bigg)
\end{equation}
\subsubsection{Step 2-Drive} The relaxation converts \eqref{equ. relaxed general optimization problem} into a convex problem by setting $\tilde{t}=0$, and the solutions are used as a warm start for the drive phase. It is proposed to adjust $\tilde{t}$ as \eqref{equ. updating t}, based on the primal and dual residuals defined in \eqref{equ. Primal and Dual residuals definition} \cite{Boyd2011Distributed}.
\begin{equation}\label{equ. updating t}
    \tilde{t}^{(k)} := \tilde{t}^{(k-1)} + c\;(1/r_p^{(k-1)}+1/r_d^{(k-1)})
\end{equation}
\begin{equation}\label{equ. Primal and Dual residuals definition}
    r_p^{(k)} = \|\vec{y}^{(k)}-\vec{z}^{(k)}\|_2, \; r_d^{(k)}=\rho\|\vec{z}^{(k)}-\vec{z}^{(k-1)}\|_2
\end{equation}\par
This ensures that the rate of changing $\tilde{t}$ corresponds to the stabilized convergence during the procedure. If residuals are small, the algorithm is stabilized and $\tilde{t}$ can be increased to further push binary variables toward Boolean values; otherwise, if residuals are large, more iterations are required and $\tilde{t}$ should not be boosted. After reaching specific iterations or residuals, the procedure might be completed by $\tilde{t} \to \infty$ or the projection of remained variables onto final Boolean values.\par
\subsubsection{Step 3-Polish} In order to verify the results, the polishing phase fixes the values of binary variables based on the results obtained from the previous phase, and solves the convex problem using the convex ADMM for the remaining variables.  \par 
To realize distributed framework and enhance the scalability, the proposed NC-ADMM method integrates the relax-drive-polish procedure with the consensus ADMM, as shown in \eqref{equ. relax-drive-polish with consensus ADMM for generalized optimization} for cluster $i$. The idea is to introduce and reparametrize consensus continuous variables as $\vec{\overline{x}}$ and consensus binary variables among each $B$ neighbor clusters. Then, $\vec{{z}}$ reappears as a consensus binary variable, being integrated with all binding and inner ones, which $B=1$ for inner variables.
\vspace{-0.4cm}
\begin{subequations} \label{equ. relax-drive-polish with consensus ADMM for generalized optimization}
 \begin{align}
 \begin{split}\label{equ. step 1 relax-drive-polish with consensus ADMM for generalized optimization}
 (\vec{x}_i,\vec{y}_i)^{(k+1)} &:= \argmin_{\vec{x_i},\vec{y_i} \in C_i,\: 0\leq\vec{y_i}\leq1} \bigg[ f_i(\vec{x}_i,\vec{y}_i)+\\ \frac{\rho}{2} \| \vec{x}_i^{Bi}-& \vec{\overline{x}}_i^{(k)} + \vec{\mathfrak{u}}_{1,i}^{(k)}\|_2^2+ \frac{\rho}{2} \| \vec{y}_i-\vec{z}_i^{(k)} +\vec{\mathfrak{u}}_{2,i}^{(k)}\|_2^2 \bigg]
 \end{split}{}\\
 \begin{split}\label{equ. step 2 relax-drive-polish with consensus ADMM for generalized optimization}
     \vec{\overline{x}}_i^{(k+1)}&:=\frac{1}{B} \sum_{i=1}^B(\vec{x}_i^{Bi^{(k+1)}}+\vec{\mathfrak{u}}_{1,i}^{(k)}) \\
     \vec{z}_i^{(k+1)}&:=\prox_{\tilde{t},I(\vec{z})}\bigg(\frac{1}{B} \sum_{i=1}^B(\vec{y}_i^{(k+1)}+\vec{\mathfrak{u}}_{2,i}^{(k)})\bigg)
 \end{split}{}\\
 \begin{split}\label{equ. step 3 relax-drive-polish with consensus ADMM for generalized optimization}
     \vec{\mathfrak{u}}_{1,i}^{(k+1)}&:=\vec{\mathfrak{u}}_{1,i}^{(k)}+\vec{x}_i^{Bi^{(k+1)}}-\vec{\overline{x}}_i^{(k+1)}\\
     \vec{\mathfrak{u}}_{2,i}^{(k+1)}&:=\vec{\mathfrak{u}}_{2,i}^{(k)}+\vec{y}_i^{(k+1)}-\vec{z}_i^{(k+1)}
 \end{split}
\end{align}
\end{subequations}

In order to calculate the residuals in the NC-ADMM method, \eqref{equ. Primal and Dual residuals definition} should incorporate the consensus continuous variables, as shown in \eqref{equ. main problem primal dual residuals}. The detailed procedure of the proposed NC-ADMM method is provided in Algorithm \ref{Alg. NC_ADMM}.
\begin{subequations}\label{equ. main problem primal dual residuals}
   \begin{equation}\label{equ. main problem primal residual}
       r_p^{(k)} = \|\vec{y}^{(k)}-\vec{z}^{(k)}\|_2 + \|\vec{x}^{(k)}-\vec{\overline{x}}^{(k)}\|_2
   \end{equation}
   \begin{equation}\label{equ. main problem dual residual}
       r_d^{(k)}=\rho\|\vec{z}^{(k)}-\vec{z}^{(k-1)}\|_2+\rho\|\vec{\overline{x}}^{(k)}-\vec{\overline{x}}^{(k-1)}\|_2
   \end{equation}
  \end{subequations}
  
\begin{algorithm}
 \caption{The proposed NC-ADMM method}
\begin{algorithmic}[1]
 \STATE \textbf{Initialization:} $ {\vec{x}}^{(0)}, {\vec{y}}^{(0)}, {\vec{z}}^{(0)}, {\vec{\overline{x}}}^{(0)}, {\vec{\mathfrak{u}}_1}^{(0)}, {\vec{\mathfrak{u}}_2}^{(0)} =0, 
 \newline k=0, \; \rho=1, \; \tilde{t}^{(0)}=0$
 \WHILE{(Not converged \hspace{-0.15cm} \OR \hspace{-0.15cm} $\tilde{t}\neq0$) \AND Not max iteration}
 \STATE $k \gets k+1$
 \IF{$\tilde{t}^{(k-1)}=0$}
 \IF{Not converged}
 \STATE $\tilde{t}^{(k)}=0$ \COMMENT{Continue the relax phase} 
 \ELSE
 \STATE \textbf{Update} $\tilde{t}^{(k)}$ using \eqref{equ. updating t} \COMMENT{Start the drive phase} 
 \STATE \COMMENT{Warm-start with results from the relax phase}
 \ENDIF
 \ELSIF{$\tilde{t}^{(k-1)} \neq 0$ \AND $k \leq$ Max Prox iterations}
 \STATE \textbf{Update} $\tilde{t}^{(k+1)}$ using \eqref{equ. updating t} \COMMENT{Continue the drive phase}
 \ELSE
 \STATE $\tilde{t}^{(k)} \to \infty$ \COMMENT{Start or continue the projection}
 \ENDIF
  \STATE \textbf{Update} $\vec{x}_i^{(k)}$ and $\vec{y}_i^{(k)}$ of each agent $i$ by solving local DSR problem of \eqref{equ. step 1 relax-drive-polish with consensus ADMM for generalized optimization} \label{AlgLine. Line Update x}
  \STATE \textbf{Broadcast messages} of $\vec{x}_i^{{Bi}^{(k)}}$ and  $\vec{y}_i^{{Bi}^{(k)}}$ by each agent $i$ to the neighboring agents and receive data from them
  \STATE \textbf{Update} consensus or relaxed binary variables of $\vec{\overline{x}}_i^{(k)}$ and $\vec{z}_i^{(k)}$ using exchanged messages by \eqref{equ. step 2 relax-drive-polish with consensus ADMM for generalized optimization}
  \STATE \textbf{Update} ${\vec{\mathfrak{u}}_1}^{(k)}$ and ${\vec{\mathfrak{u}}_2}^{(k)}$ using \eqref{equ. step 3 relax-drive-polish with consensus ADMM for generalized optimization} 
  \STATE \textbf{Calculate} residuals by \eqref{equ. main problem primal dual residuals}
  \STATE \textbf{Check} convergence by $r_p^{(k)}\leq \epsilon$ \AND $r_d^{(k)} \leq \epsilon$ \label{AlgLine. Line Check convergence}
  \ENDWHILE \\
  \COMMENT{Start the polish phase while Boolean variables are fixed} 
  \STATE \textbf{Initialize} all variables by previous stage 
  \WHILE{Not converged}
  \STATE $k \gets k+1$
  \STATE Repeat lines \ref{AlgLine. Line Update x}-\ref{AlgLine. Line Check convergence} (Boolean variables are restricted)
  \ENDWHILE
\end{algorithmic}
\label{Alg. NC_ADMM}
\end{algorithm}
\vspace{-0.6cm}
\subsection{Application of NC-ADMM to DDSR}
The DN is clustered based on the clustering strategy in Algorithm \ref{Alg. Bottom-up travers}, and the DDSR is solved using the NC-ADMM method in Algorithm \ref{Alg. NC_ADMM}. For each cluster $i$, decision variables of problem \eqref{equ. objective function} can be decomposed as $\vec{x}_i=[\vec{x}_{i,t}], \: \vec{y}_i=[\vec{y}_{i,t}], \: \vec{\overline{x}}_i=[\vec{\overline{x}}_{i,t}], \: \vec{z}_i=[\vec{z}_{i,t}], \: \forall i \in \mathcal{K}, t \in \mathcal{T}$, as defined in \eqref{equ. DSR variables adopted for NC-ADMM}. The continuous consensus variables of \eqref{equ. DSR variables adopted for NC-ADMM x_bar variable} are derived from the power flow equations, which each cluster also considers the voltage of boundary nodes and the power flow of joint lines from neighboring clusters.
\begin{subequations}\label{equ. DSR variables adopted for NC-ADMM}
\begin{equation}\label{equ. DSR variables adopted for NC-ADMM x variable}
    \vec{x}_{i,t} := \hspace{-0.1cm}[\vec{P}_{j,t}^{L,i},\vec{Q}_{j,t}^{L,i}, \vec{U}_{j,t}^{i}, \vec{P}_{jl,t}^i, \vec{Q}_{jl,t}^i, \vec{P}_{j,t}^{G,i}, \vec{Q}_{j,t}^{G,i}, \vec{Q}_{j,t}^{\text{cap},i} ]
\end{equation}
\begin{equation}\label{equ. DSR variables adopted for NC-ADMM y variable}
\vec{y}_{i,t}=\vec{z}_{i,t}:=[x_{j,t}^{L,i},x_{j,t}^{B,i}, \alpha_{jl,t}^{i}, \beta_{jl,t}^{i}]
\end{equation}
 \begin{equation}\label{equ. DSR variables adopted for NC-ADMM x_bar variable}
     \vec{\overline{x}}_{i,t}:=[\vec{U}_{j,t}^{i}, \vec{P}_{jl,t}^i, \vec{Q}_{jl,t}^i]
 \end{equation}
\end{subequations}

Each SLA$_i$ solves a subproblem of \eqref{equ. step 1 relax-drive-polish with consensus ADMM for generalized optimization}, in which the related convex set of $C_i$ is defined by \eqref{equ. DSR non-dispatchable load pickup}-\eqref{equ. Spanning Tree constraints} within each cluster. Then, as shown in Fig. \ref{Fig. Data exchange}, clusters exchange consensus variables and Lagrange multipliers consecutively to update them as \eqref{equ. step 2 relax-drive-polish with consensus ADMM for generalized optimization} and \eqref{equ. step 3 relax-drive-polish with consensus ADMM for generalized optimization}. The exchanged data include 1) all binding continuous variables and related Lagrange multipliers; and 2) consensus binary variables, such as the status of switchable joint lines, their related auxiliary binary variables, and their associated Lagrange multipliers. In this process, all SLAs need to comply with the proposed NC-ADMM algorithm in each of the relax-drive-polish phases during convergence.
\begin{remark}\label{Remark binding binary varibles.}
$\alpha_{jl,t}^{i}$ and $\beta_{jl,t}^{i}$ are the only common consensus binary variables for the joint distribution lines, from power flow equation \eqref{equ. DSR Power flow equations} and spanning tree constraints \eqref{equ. Spanning Tree constraints}.
\end{remark}
\begin{definition}\label{Def. parent and chldren clusters}
Parent cluster is defined as the unique closest neighboring cluster to the substation or the root node, while the others are defined as children clusters. 
\end{definition}
\begin{figure}
\centering
	\includegraphics[width=2.8in]{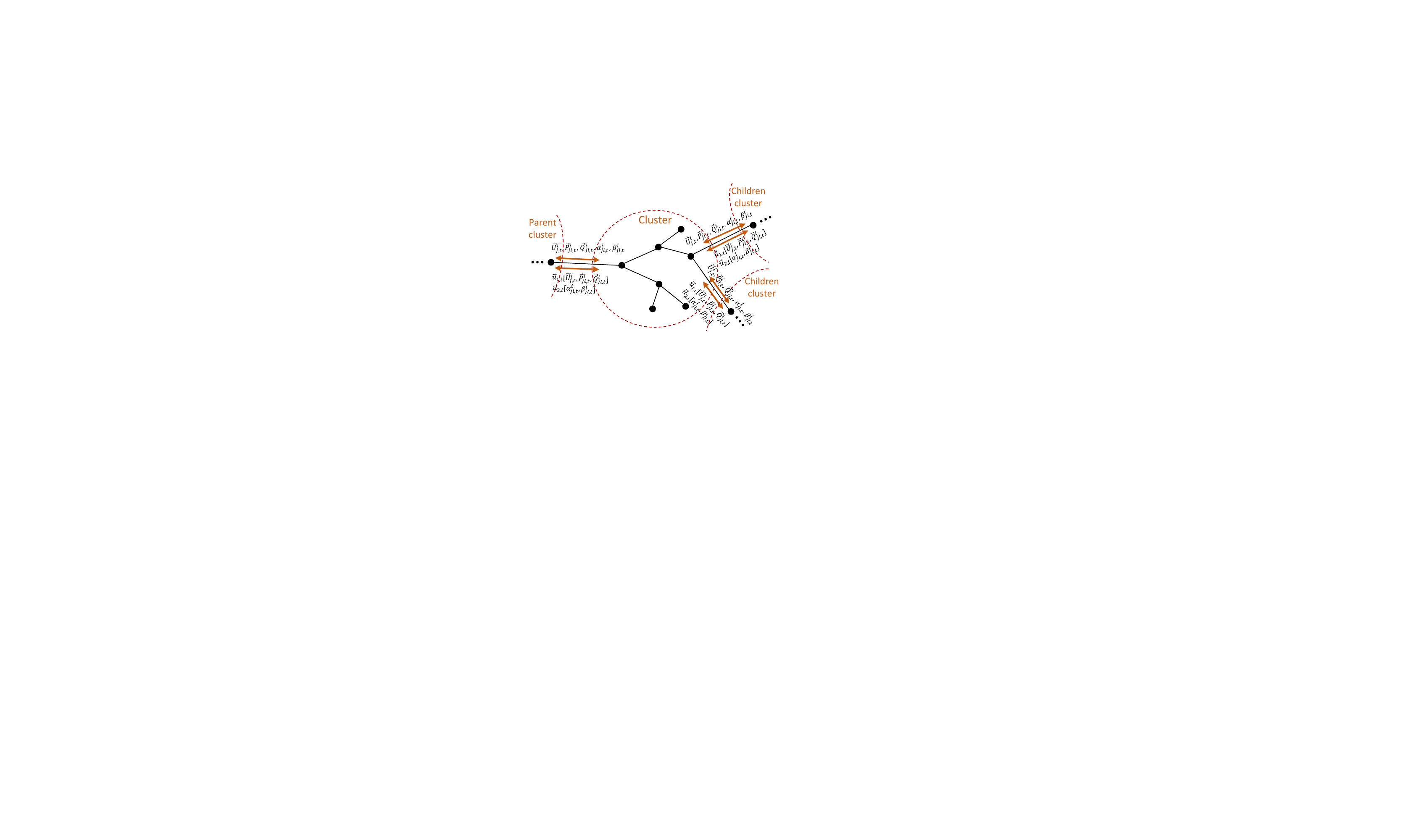}
	\caption{Data exchange for updating consensus variables.}
	\label{Fig. Data exchange}
\end{figure}
\vspace{-0.5cm}
\section{Numerical Results}\label{numerical results}
The performance and scalability of NC-ADMM-based DDSR are demonstrated though testing on IEEE123-node and IEEE 8500-node test feeders. Each time step is assumed to be 15 minutes (1 p.u.) for operation and stabilization \cite{Roofegari2019Distributed}. Simulations are implemented on Python, using Gurobi as solver and OpenDSS through COM interface as data provider.\par
\vspace{-0.5cm}

\subsection{NC-ADMM-based DDSR for Unbalanced Network}
IEEE 123-node DN is modified by connecting two inverter-based PV units at nodes 66 and 105, each with the maximum capacity of 300kW. Loads at nodes 48 and 65 have higher priority, and loads at nodes 47 and 76 are dispatchable. There are total 3 time steps following transmission restoration, which reflect the gradually increasing generation capabilities of the bulk system as [400kW, 1400kW, 3500kW] \cite{Golshani2018Coordination}. It is assumed that all loads have a constant power factor of 0.9, and the nodal voltage is limited within 0.95 pu and 1.05 pu. 

Based on the proposed clustering strategy, the DN is divided into 11 clusters, as shown in Fig. \ref{Fig. DDSR operation scenario 1, IEEE 123-node}. Two outage scenarios are considered in which scenario 1 represents a restoration after blackout and scenario 2 combines scenario 1 with emergency restoration. The ADMM parameter $\rho$ is set as 1.\par

\subsubsection{Scenario 1}
It shows how the proposed method deals with a network with outages after a blackout. The provided power through substation gradually increases, and PV generators are operated to pick up more loads. Fig. \ref{Fig. DDSR operation scenario 1, IEEE 123-node} shows the status of switching operation and load pickup for the first time step of restoration. All energized and de-energized loads are shown by green and red dots, respectively. Blue downside and yellow upside arrows show closed and open switch during restoration. As there is not any faulted line in the network, all normally-open switches remain open to prevent any loop during the network operation. Clearly, all high priority loads have been picked up during the first time step of DDSR.\par  

\begin{figure}
\centering
	\includegraphics[width=3.2in]{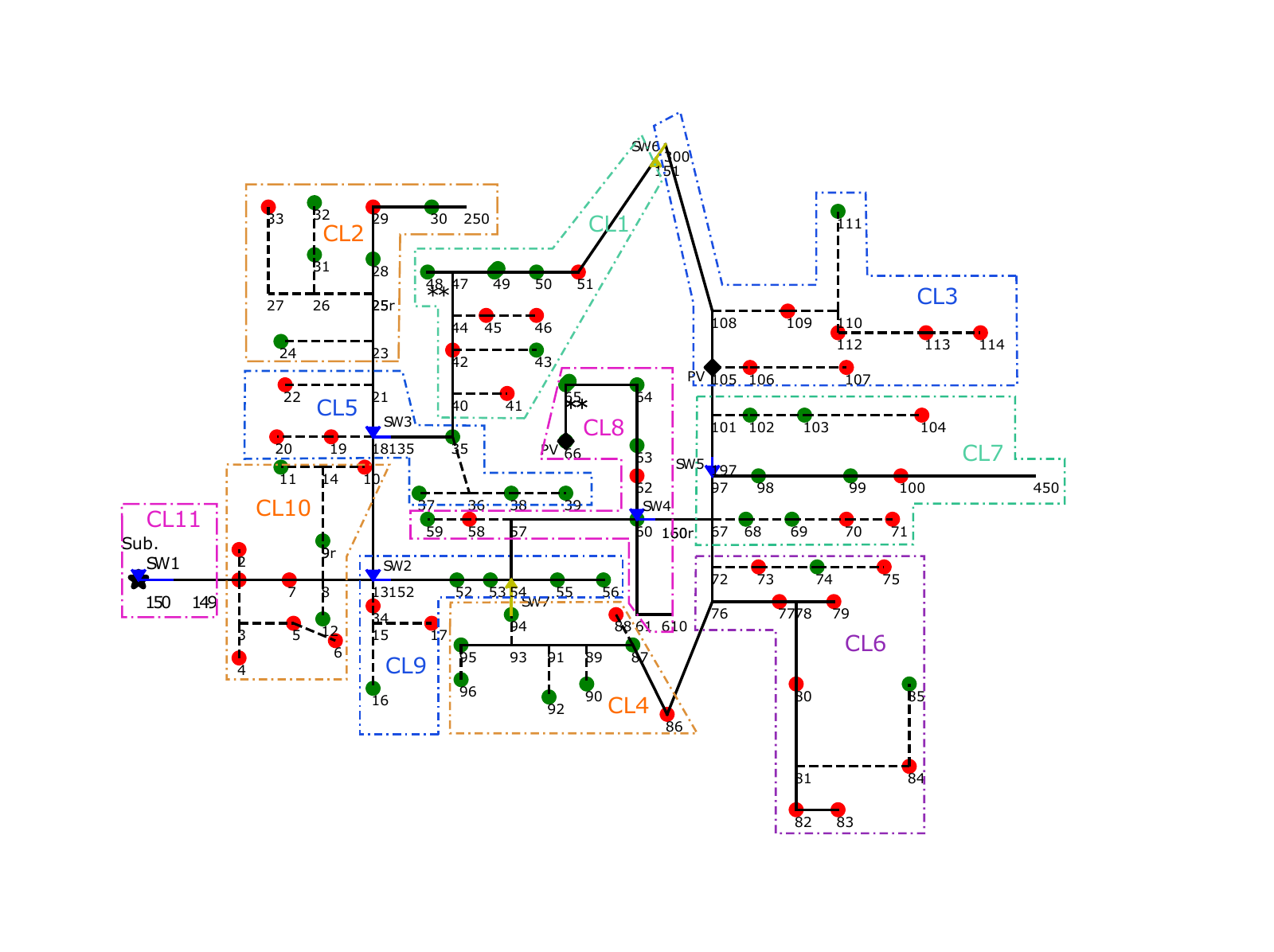}
	\caption{Clustering results and DDSR operation in first time step of scenario 1 for IEEE 123-node DN without faulted line.}
	\label{Fig. DDSR operation scenario 1, IEEE 123-node}
\end{figure}

\subsubsection{Scenario 2}
Considering previous scenario, there is also a faulted line between nodes 72 and 76, which is isolated through opening switches 4, 5 and 7 as shown in Fig. \ref{Fig. DDSR operation scenario 2, IEEE 123-node}. PV generators are dispatched to provide more power for the load pickup. Due to space limitations, it only shows the first time step of DDSR in which all high priority loads have been picked up. Furthermore, the normally-open switch 6 is closed to provide power to the unfaulted out-of-service area without any loop in operation.\par

\begin{figure}
\centering
	\includegraphics[width=3.2in]{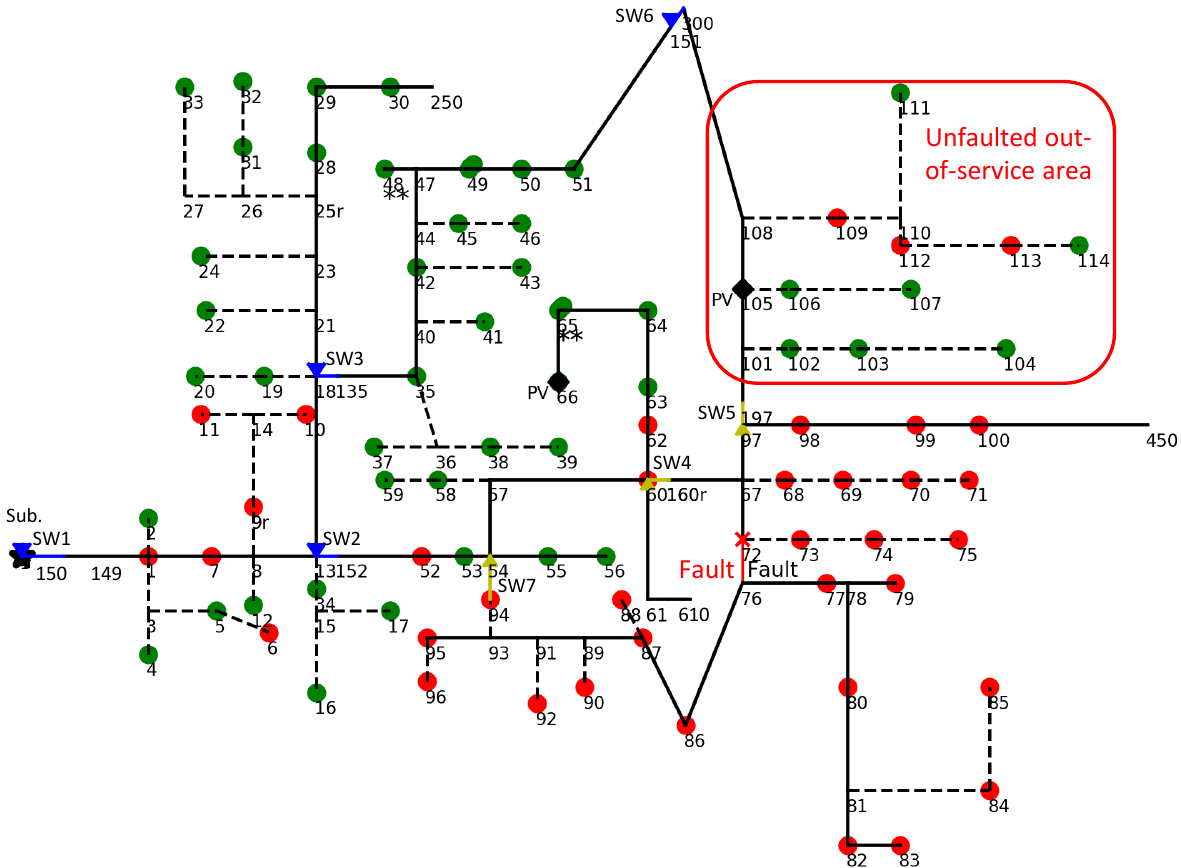}
	\caption{DDSR operation in first time step of scenario 2 for IEEE 123-node DN with a faulted line.}
	\label{Fig. DDSR operation scenario 2, IEEE 123-node}
\end{figure}
\vspace{-0.4cm}
\subsection{Performance of NC-ADMM Method}
The convergence of the NC-ADMM-based DDSR, in terms of total restored loads for each cluster in scenario 1, is shown in Fig. \ref{Fig. clusters load convergence SC1}. The stopping criteria for drive and polish phases are the primal and dual residuals less than $\sqrt{\text{\# of Agents}}\times10^{-4}$, which is increased to be 10 times larger for relax phase. It shows that, 1) during the relax phase, each SLA picks all of its loads; 2) by exchanging data among clusters and adjusting Lagrange multipliers, power flow equations start to affect, and the amount of load pickup drops following the total available power; 3) during the drive phase, the load pickup amount are driven toward Boolean values; 4) at the end of drive phase, by $\tilde{t} \to \infty$, all remaining binary variables are projected to the nearest Boolean values; 5) during the polish phase, binary variables are fixed and continuous variables such as dispatchable loads converge to better optimal values.\par 

\begin{figure}
\centering
	\includegraphics[width=3.2in]{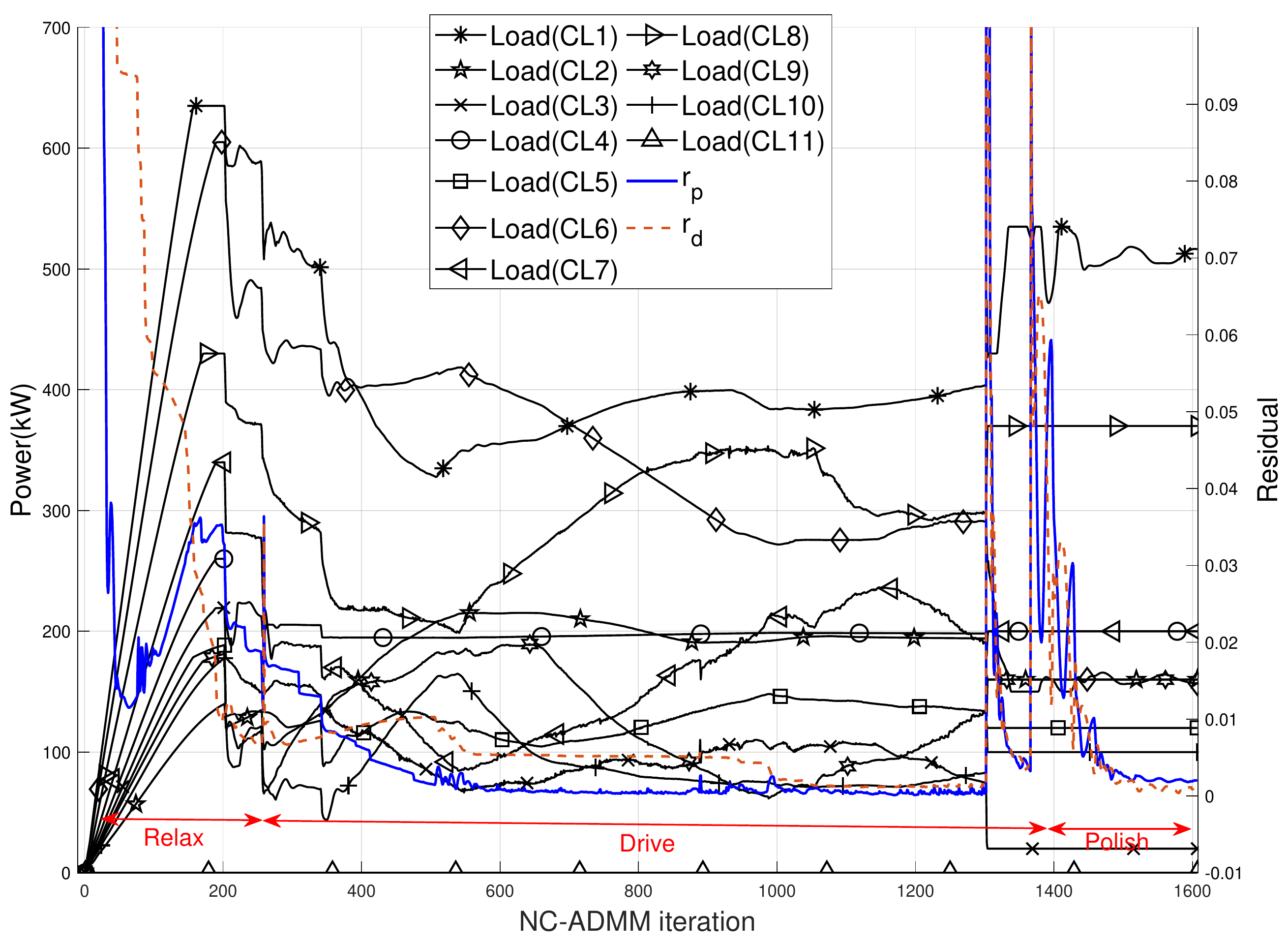}
	\caption{Total load convergence of each cluster at first time step and primal and dual residuals in scenario 1.}
	\label{Fig. clusters load convergence SC1}
\end{figure}

Fig. \ref{Fig. Switching convergence SC1} shows the convergence of switching status, which are correctly converged during the relax phase. During the drive phase, despite stimulation of changing and inspecting other values, they are reverted immediately since the initial values are optimal, as all switches remain their original status. It also shows the evolution of $\tilde{t}$ during the convergence, as it equals to zero during the relax phase, and constantly increases to larger values according to \eqref{equ. updating t}, by setting $c=10^{-1}$.\par
\begin{figure}
\centering
	\includegraphics[width=3.2in]{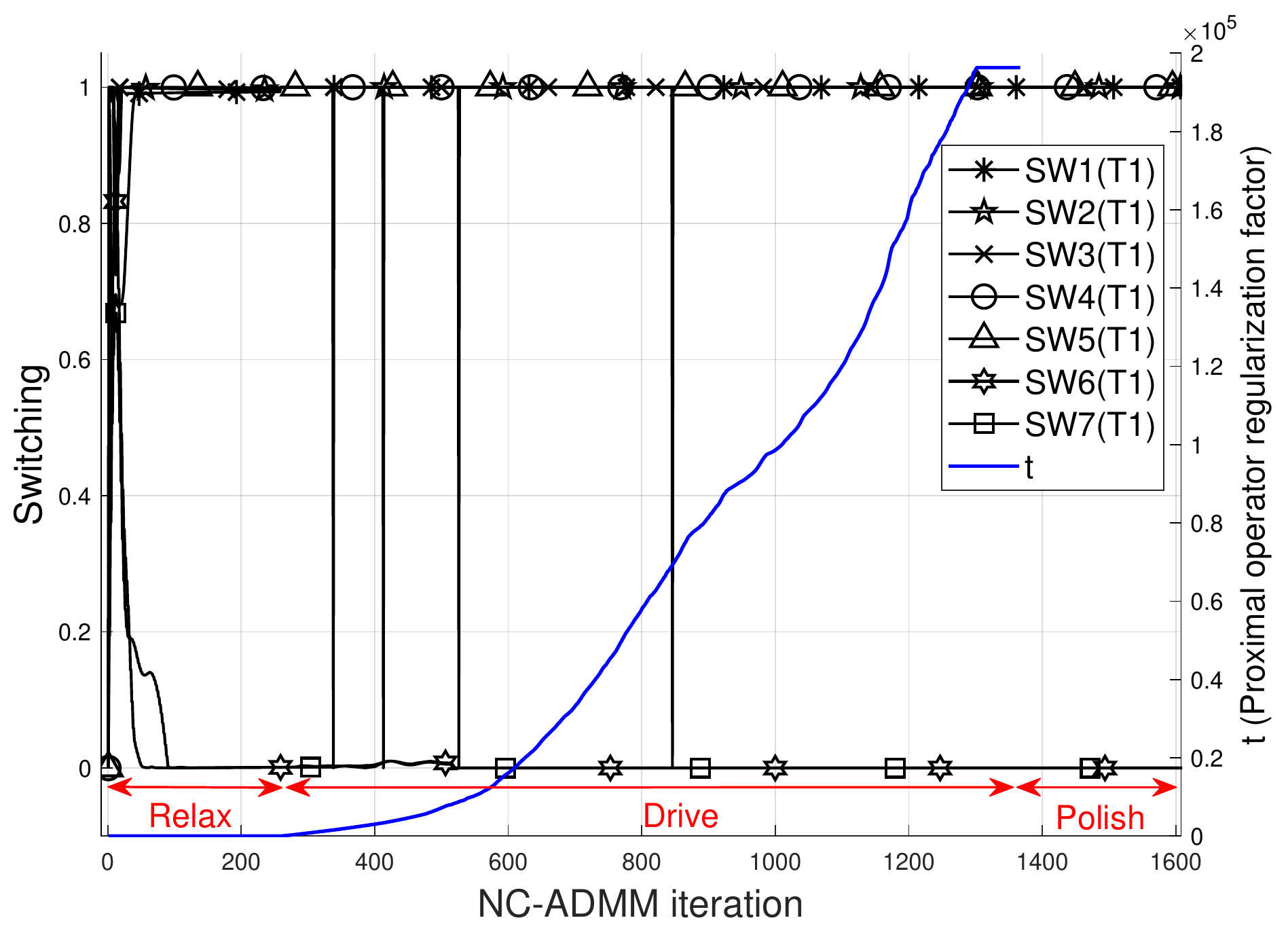}
	\caption{Switching status and proximal operator regularization factor convergences in scenario 1.}
	\label{Fig. Switching convergence SC1}
\end{figure}

Fig. \ref{Fig. Projected Residuals SC1} shows the residuals using the conventional heuristic projection method in \eqref{equ. generalized projection Non-convex ADMM}, and the comparison with the proposed NC-ADMM method is shown in Fig. \ref{Fig. Restored Loads scenario 1, IEEE 123-node}. It is shown that the conventional method of \eqref{equ. generalized projection Non-convex ADMM} oscillates even after a large number of iterations. If the algorithm halted, the final values of conventional method \eqref{equ. generalized projection Non-convex ADMM} are usually infeasible, as the second time step in Fig. \ref{Fig. Restored Loads scenario 1, IEEE 123-node}, due to the false switching operation.\par

\begin{figure}
\centering
	\includegraphics[width=2.6in]{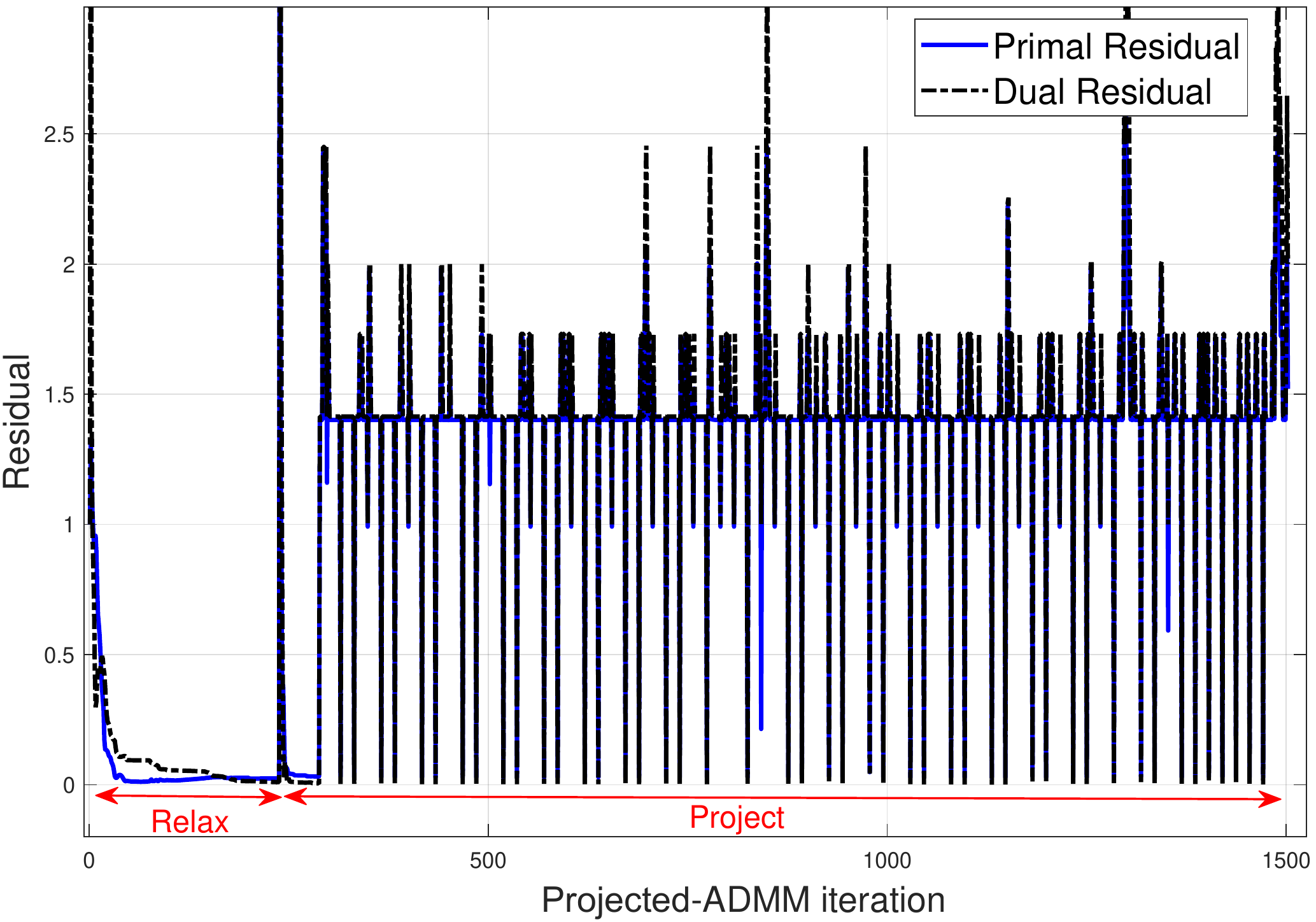}
	\caption{Primal and dual residuals of the projection method \eqref{equ. generalized projection Non-convex ADMM}.}
	\label{Fig. Projected Residuals SC1}
\end{figure}

\begin{figure}
\centering
	\includegraphics[width=2.6in]{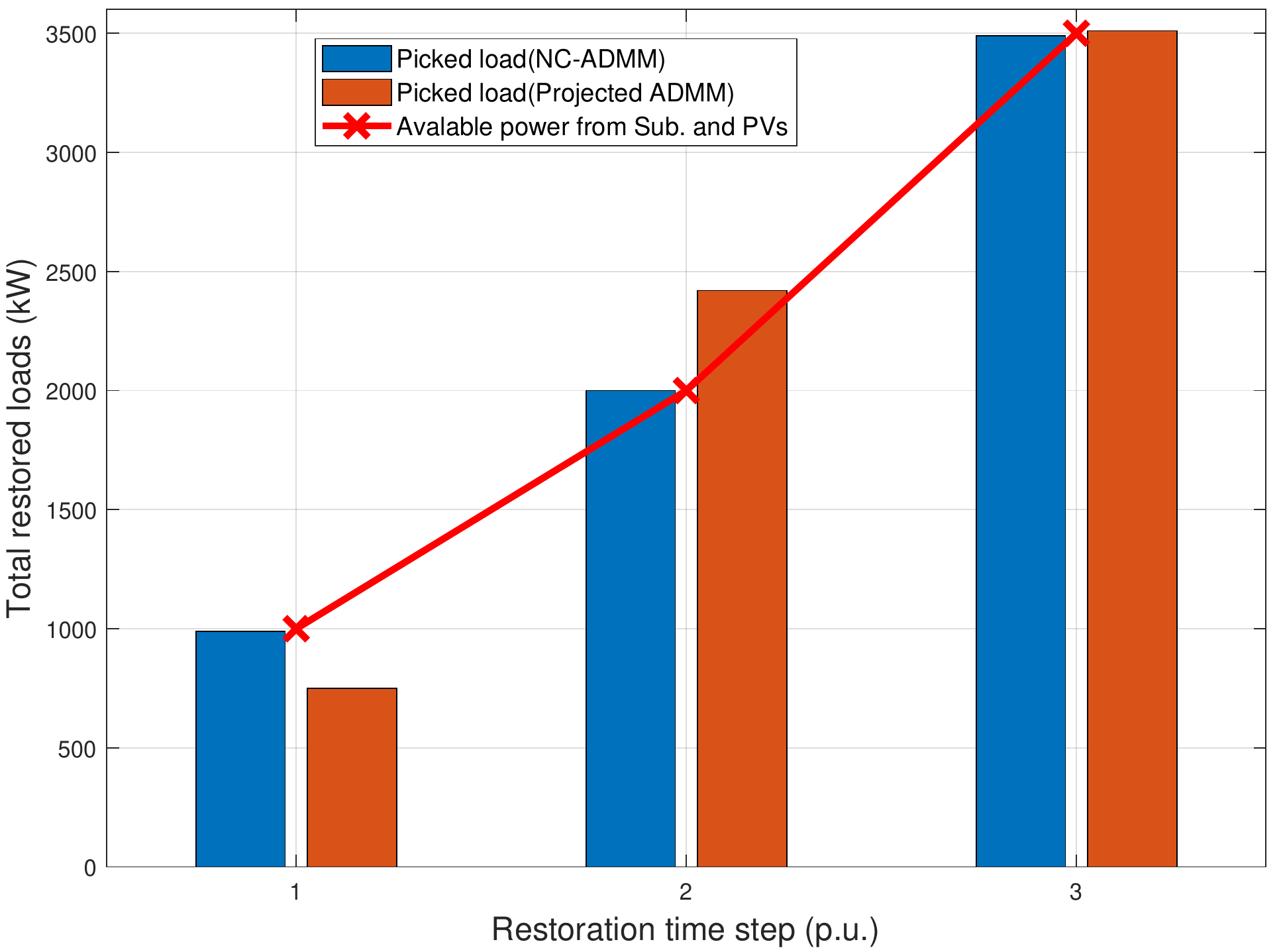}
	\caption{Comparison of total restored loads in scenario 1 between the NC-ADMM method and the method \eqref{equ. generalized projection Non-convex ADMM}.}
	\label{Fig. Restored Loads scenario 1, IEEE 123-node}
\end{figure}
\begin{figure}[!]
\centering
	\includegraphics[width=3.2in]{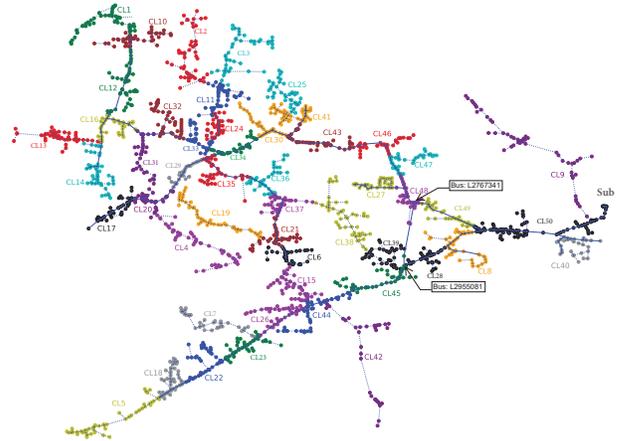}
	\caption{Clustering of IEEE 8500-node DN for DDSR.}
	\label{Fig. clustering IEEE 8500-Node}
\end{figure}

\vspace{-0.3cm}
\subsection{Scalability of the NC-ADMM-based DDSR}
IEEE 8500-node network is integrated with two 1,000 kW PV generators on buses `1026706' and `1047592'. The circuit consists of 35 normally-closed switches and 5 tie-switches, and a long switchable tie-line is added between buses `L2767341' and `L2955081'. This network consists of 2,522 primary buses while the secondary networks and loads are aggregated into the related secondary transformer buses. The first 100 loads in the related load document are considered as dispatchable, and the rest are binary-valued. There are total 3 restoration steps with the power capacity of [600kW, 6,000kW, 12,000kW] following transmission restoration. \par

Using the proposed clustering strategy, the network is divided into 50 clusters as shown in Fig. \ref{Fig. clustering IEEE 8500-Node}. A blackout restoration is assumed with a faulted line between buses `M1125934' and `L2730163', which is isolated through sectionalizing switches `A8645\_48332\_sw' and `A8611\_48332\_sw'. Fig. \ref{Fig. DDSR operation, IEEE 8500-node} shows the last time step of restoration. The long tie-line switch is closed to provide power for the out-of-service area after the faulted line. It is clear that all other tie-switches are open as shown by a yellow arrow, to prevent any loop in the operation. Furthermore, the convergence of the proposed method, in terms of total restored loads at each time step, is shown in Fig \ref{Fig. Total load convergence for 8500-node}. Similarly, SLAs solve the problem in each phase and exchange data among each other, until the related Lagrange multipliers for the power balance affected and the total available power adjusted during the drive phase. During the polish phase, all binary variables are fixed to polish the results toward high-quality solution for the DDSR problem.   \par

\begin{figure}
\centering
	\includegraphics[width=3.1in]{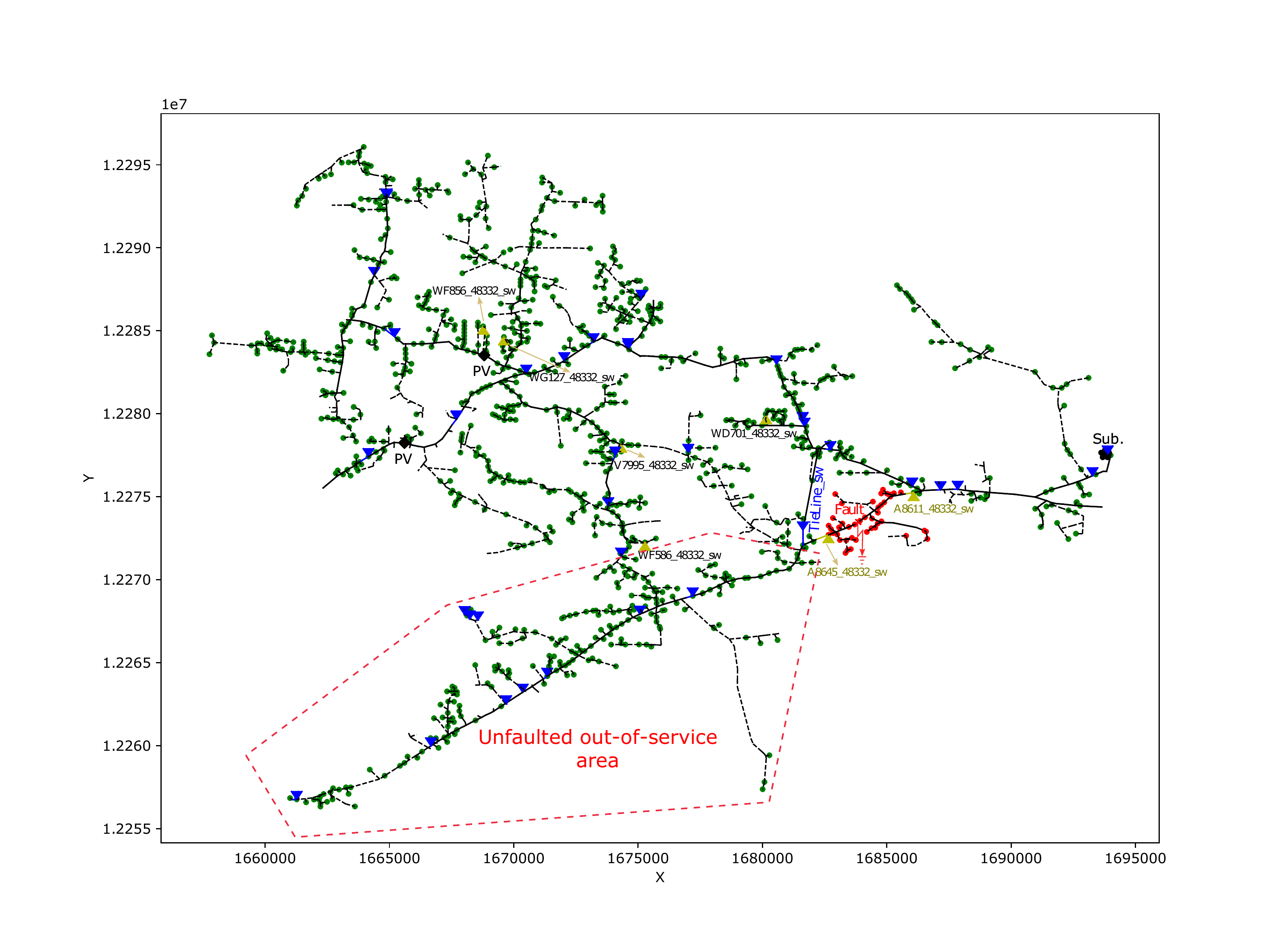}
	\caption{DDSR operation in last time step for IEEE 8500-node DN with a faulted line.}
	\label{Fig. DDSR operation, IEEE 8500-node}
\end{figure}

\begin{figure}
\centering
	\includegraphics[width=3in]{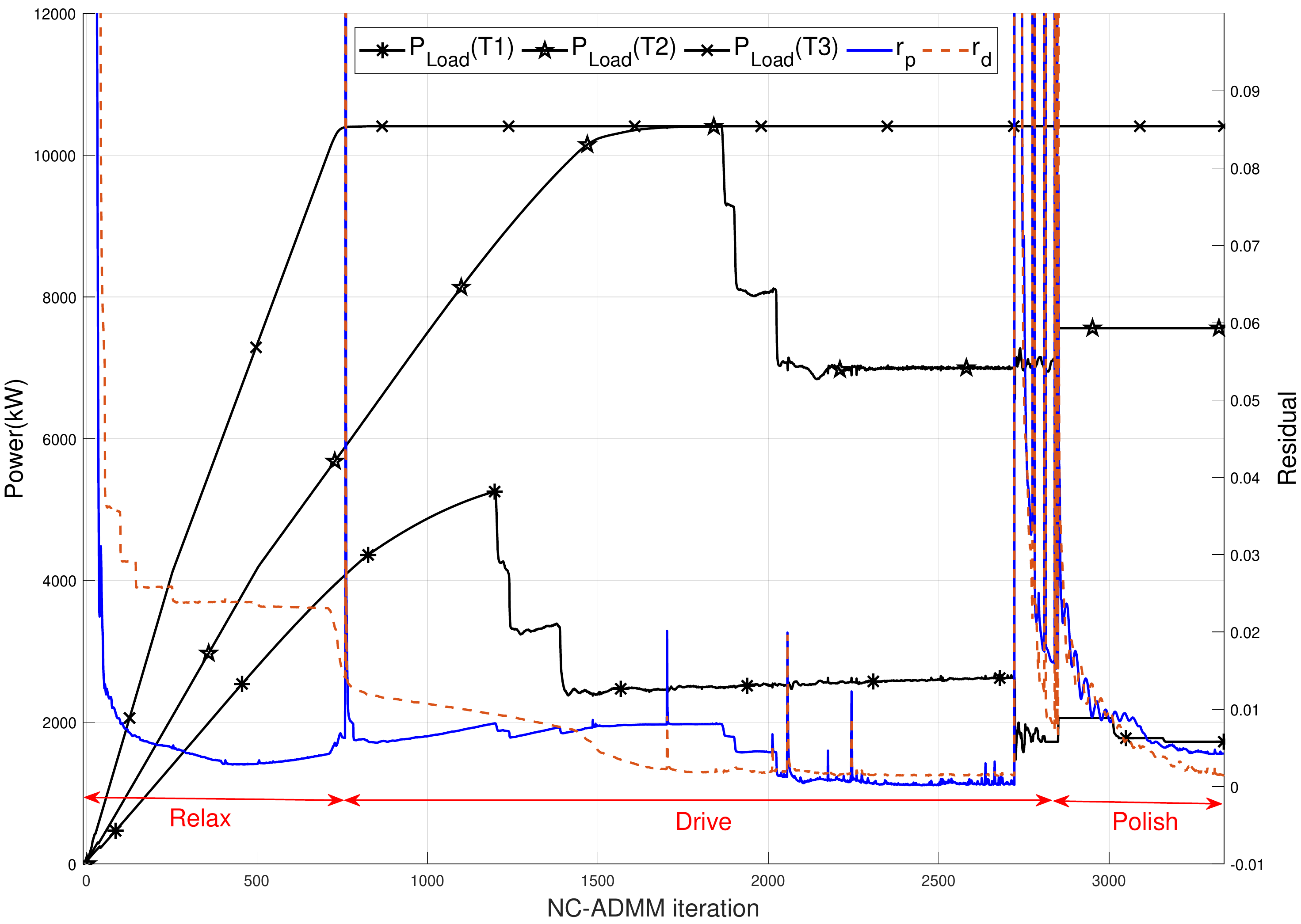}
	\caption{Total load convergence for each time step and residuals for IEEE 8500-node DN.}
	\label{Fig. Total load convergence for 8500-node}
\end{figure}

\section{Conclusion}
Service restoration can be formulated as a challenging mixed-integer nonlinear programming problem, which deals with many binary variables representing loads and switching operation. In this paper, a non-convex ADMM-based distributed optimization method is developed and applied to the service restoration problem in large-scale active distribution networks. The developed heuristic ADMM-based algorithm is incorporated with consensus ADMM to provide a fully distributed cluster-based framework. Moreover, an adaptive autonomous clustering strategy is developed for application in large-scale networks, in which each cluster consists of a smart agent to carry out the distributed restoration procedure with limited data exchange with neighbors. Simulation results on large-scale IEEE test networks demonstrate the capability of the distributed restoration to deal with various blackout or emergency restoration problems, and also the superiority of the distributed non-convex method over simple projection methods. In future work, the proposed NC-ADMM method can be further analyzed in terms of the evolution of parameter $t$ during the drive phase to achieve an even better solution.\par 


%





\ifCLASSOPTIONcaptionsoff
  \newpage
\fi



%

\bibliographystyle{IEEEtran}
\bibliography{mybib}

%








\end{document}